\def\year{2020}\relax
\documentclass[letterpaper]{article} 
\usepackage{aaai20}  
\usepackage{times}  
\usepackage{helvet} 
\usepackage{courier}  
\usepackage[hyphens]{url}  
\usepackage{graphicx} 
\usepackage{graphicx}  
\usepackage{amsfonts,amsthm,amsmath}
\usepackage{xspace}
\usepackage{algpseudocode}
\usepackage[english]{babel}
\usepackage[utf8]{inputenc}

\newcommand{\G}{{\cal G}}

\newcommand{\Q}{\mathbb{Q}}
\newcommand{\Qpos}{\mathbb{Q}_{>0}}
\newcommand{\Nat}{\mbox{I$\!$N}}
\newcommand{\Real}{\mbox{I$\!$R}}

\newcommand{\SThresh}{\textsc{SThr}}

\newcommand{\PO}{Player~$1$\xspace}
\newcommand{\PT}{Player~$2$\xspace}
\newcommand{\PLi}{Player~$i$\xspace}
\newcommand{\stam}[1]{}
\newcommand{\zug}[1]{\langle #1  \rangle}
\newcommand{\set}[1]{\{ #1 \}}
\newcommand{\WnR}{\text{WnR}}
\newcommand{\deps}{d_{\epsilon}}

\newtheorem{theorem}{Theorem}[section]

\newtheorem{lemma}[theorem]{Lemma}

\newtheorem{definition}{Definition}[section]

\theoremstyle{definition}
\newtheorem{remark}{Remark}[section]
\newtheorem{example}{Example}[section]

\title{All-Pay Bidding Games on Graphs}

\author{Guy Avni\textsuperscript{\rm 1}, Rasmus Ibsen-Jensen\textsuperscript{\rm 2}, and Josef Tkadlec\textsuperscript{\rm 1}\\
\textsuperscript{\rm 1}IST Austria\\
\textsuperscript{\rm 2} Liverpool University}

\stam{\author{Written by AAAI Press Staff\textsuperscript{\rm 1}\thanks{Primarily Mike Hamilton of the Live Oak Press, LLC, with help from the AAAI Publications Committee}\\ \Large \textbf{AAAI Style Contributions by
Pater Patel Schneider,} \\ \Large \textbf{Sunil Issar, J. Scott Penberthy, George Ferguson, Hans Guesgen}\\ 
\textsuperscript{\rm 1}Association for the Advancement of Artificial Intelligence\\ 
2275 East Bayshore Road, Suite 160\\
Palo Alto, California 94303\\
publications20@aaai.org 
}}

\begin{document}
\maketitle
\begin{abstract}
In this paper we introduce and study {\em all-pay bidding games}, a class of two player, zero-sum games on graphs. The game proceeds as follows. We place a token on some vertex in the graph and assign budgets to the two players. Each turn, each player submits a sealed legal bid (non-negative and below their remaining budget), which is deducted from their budget and the highest bidder moves the token onto an adjacent vertex. The game ends once a sink is reached, and \PO pays \PT the outcome that is associated with the sink. The players attempt to maximize their expected outcome. Our games model settings where effort (of no inherent value) needs to be invested in an ongoing and stateful manner. On the negative side, we show that even in simple games on DAGs, optimal strategies may require a distribution over bids with infinite support. A central quantity in bidding games is the {\em ratio} of the players budgets. On the positive side, we show a simple FPTAS for DAGs, that, for each budget ratio, outputs an approximation for the optimal strategy for that ratio. We also implement it, show that it performs well, and suggests interesting properties of these games. Then, given an outcome $c$, we show an algorithm for finding the necessary and sufficient initial ratio for guaranteeing outcome $c$ with probability~$1$ and a strategy ensuring such. Finally, while the general case has not previously been studied, solving the specific game in which \PO wins iff he wins the first two auctions, has been long stated as an open question, which we solve.
\end{abstract}

\section{Introduction}
Two-player {\em graph games} naturally model settings in which decision making is carried out dynamically. Vertices model the possible configurations and edges model actions. The game proceeds by placing a token on one of the vertices and allowing the players to repeatedly move it. One player models the protagonist for which we are interested in finding an optimal decision-making strategy, and the other player, the antagonist, models, in an adversarial manner, the other elements of the systems on which we have no control. 

We focus on {\em quantitative reachability} games \cite{Ev55} in which the graph has a collection of sink vertices, which we call the {\em leaves}, each of which is associated with a weight. The game is a zero-sum game; it ends once a leaf is reached and the weight of the leaf is \PO's reward and \PT's cost, thus \PO aims at maximizing the weight while \PT aims at minimizing it. A special case is {\em qualitative reachability} games in which each \PLi has a {\em target} $t_i$ and \PLi {\em \/ wins} iff the game ends in $t_i$.

A graph game is equipped with a mechanism that determines how the token is moved; e.g., in {\em turn-based} games the players alternate turns in moving the token. {\em Bidding} is a mode of moving in which in each turn, we hold an auction to determine which player moves the token. Bidding qualitative-reachability games where  studied in \cite{LLPU96,LLPSU99} largely with variants of first-price auctions: in each turn both players simultaneously submit bids, where a bid is legal if it does not exceed the available budget, the higher bidder moves the token, and pays his bid to the lower bidder in {\em Richman bidding} (named after David Richman), and to the bank in {\em poorman bidding}. The central quantity in these games is the {\em ratio} between the players' budgets. Each vertex is shown to have a {\em threshold ratio}, which is a necessary and sufficient initial ratio that guarantees winning the game. Moreover, optimal strategies are deterministic. 

We study, for the first time, quantitative-reachability {\em all-pay} bidding games, which are similar to the bidding rules above except that both players pay their bid to the bank. Formally, for $i \in \set{1,2}$, suppose that \PLi's budget is $B_i \in \Qpos$ and that his bid is $b_i \in [0,B_i]$, then the higher bidder moves the token and \PLi's  budget  is updated to $B_i-b_i$. Note that in variants of first-price auctions, assuming the winner bids $b$, the loser's budget is the same for any bid in $[0,b)$. In an all-pay auction, however, the higher the losing bid, the lower the loser's available budget is in the next round. Thus, intuitively, the loser would prefer to bid as close as possible to $0$. 

\begin{figure}[t]
\begin{minipage}[b]{0.4\linewidth}
\centering
\includegraphics[height=1.5cm]{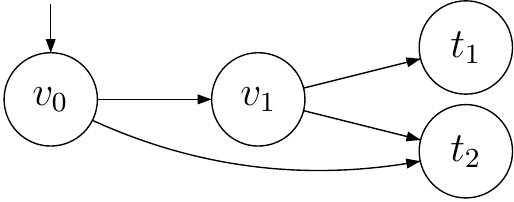}
\end{minipage}
\hspace{0.05\linewidth}
\begin{minipage}[b]{0.51\linewidth}
\centering
\includegraphics[height=3cm]{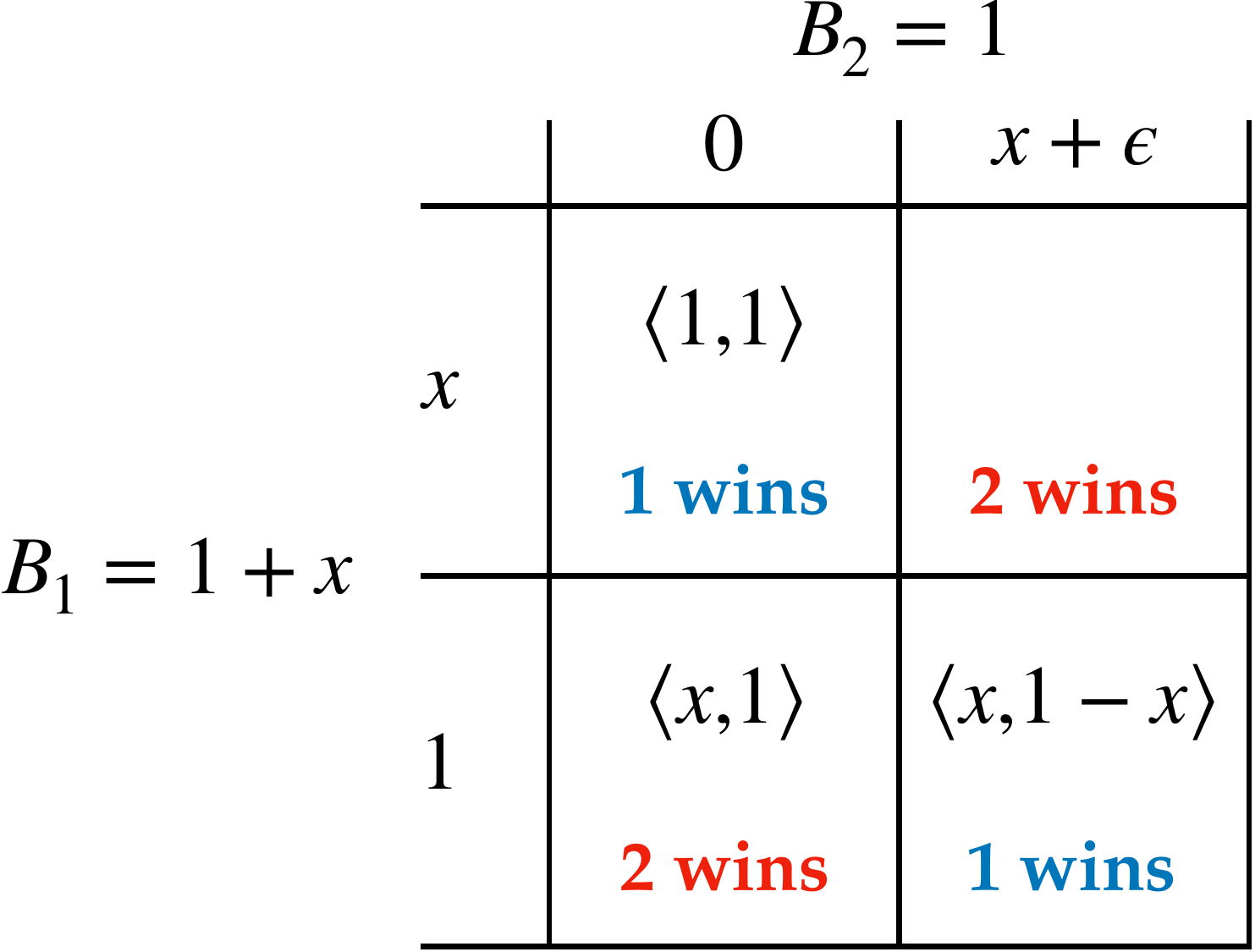}
\end{minipage}
\caption{On the left, the bidding game in which \PO wins iff he wins two biddings in a row. Assume initial budgets of $B_1 = 1+x$, for $x \in (0.5,1)$, and $B_2 = 1$. A \PO strategy that guarantees value $0.5$ uniformly bids $\set{x, 1}$ in the first bidding. \PT has two optimal counter strategies, and the table on the right depicts the budget in the second bidding in all cases.}
\label{fig:2-1}
\end{figure}

\stam{
\begin{figure}[ht]
\center
\includegraphics[height=1.5cm]{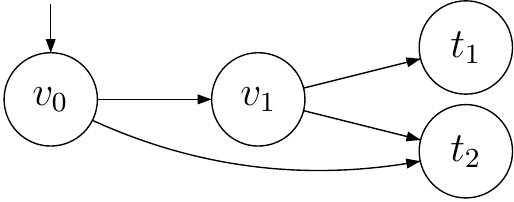}
\caption{A bidding game in which \PO wins iff the game ends in $t_1$, i.e., he wins two biddings in a row.}
\label{fig:2-1}
\end{figure}
}

\begin{example}
\label{ex:2-1}
Consider the qualitative reachability game that is depicted in Fig.~\ref{fig:2-1}, which we call {\em win twice in a row} or $\WnR(2)$, for short. For convenience, fix \PT's initial budget to be $1$. The solution to the game using a first-price auction is trivial: for example, with poorman bidding (in which the winner pays the bank), \PO wins iff his budget exceeds $2$. Indeed, if his budget is $2+\epsilon$, he bids $1+\epsilon/2$ in the first bidding, moves the token to $v_1$, and, since in first-price auctions, the loser does not pay his bid, the budgets in the next round are $1+\epsilon/2$ and $1$, respectively. \PO again bids $1+\epsilon/2$, wins the bidding, moves the token to $t_1$, and wins the game. On the other hand, if \PO's budget is $2-\epsilon$, \PT can guarantee that the token reaches $t_2$ by bidding $1$ in both rounds. 

The solution to this game with all-pay bidding is much more complicated and was posed as an open question in \cite{LLPSU99}, which we completely solve. Assuming \PT's budget is $1$, it is easy to show that when \PO's initial budget is either greater than $2$ or smaller than $1$, there exist a deterministic winning strategy for one of the players. Also, for budgets in between, i.e, in $(1, 2)$, it is not hard to see that optimal strategies require probabilistic choices, which is the source of the difficulty and an immediate difference from first-price bidding rules. We characterize the {\em value} of the game, which is the optimal probability that \PO can guarantee winning, as a function of \PO's initial budget. In Theorem~\ref{thm:2-1}, we show that for $n \in \Nat$ and $x \in (1/(n+1), 1/n]$, the value is $1/(n+1)$ when \PO's  initial budget is $1+x$ and \PT's initial budget is $1$. Fig.~\ref{fig:2-1} gives a flavor of the the solution in the simplest interesting case of $x \in (0.5, 1)$, where a \PO strategy that bids $x$ and $1$ each with probability $0.5$ guarantees winning with probability $0.5$.\hfill\qed
\end{example}

Apart from the theoretical interest in all-pay bidding games, we argue that they are useful in practice, and that they address limitations of previously studied models. All-pay auctions are one of the most well-studied auction mechanisms \cite{BKV96}. Even though they are described in economic terms, they are often used to model settings in which the agents' bids represent the amount of effort they invest in a task such as in rent seeking \cite{Tul80}, patent races, e.g., \cite{BH03}, or biological auctions \cite{CRN12}. As argued in \cite{KK09}, however, many decision-making settings, including the examples above, are not one-shot in nature, rather they develop dynamically. Dynamic all-pay auctions have been used to analyze, for example, political campaigning in the USA \cite{KP06}, patent races \cite{HV85}, and \cite{KP06} argue that they appropriately model sport competitions between two teams such as ``best of $7$'' in the NBA playoffs. An inherent difference between all-pay repeated auctions and all-pay bidding games is that our model assumes that the players' effort has no or negligible inherent value and that it is bounded. The payoff is obtained only from the reward in the leaves. For example, a ``best of $k$'' sport competition between two sport teams can be modelled as an all-pay bidding game as follows. A team's budget models the sum of players' strengths. A larger budget represents a fresher team with a deeper bench. Each bidding represents a match between the teams, and the team that bids higher wins the match. The teams only care about winning the tournament and the players' strengths have no value after the tournament is over. 

The closest model in spirit is called {\em Colonel Blotto games}, which dates back to \cite{Bor21}, and has been extensively studied since. In these games, two colonels own armies and compete in $n$ battlefields. Colonel Blotto is a one-shot game: on the day of the battle, the two colonels need to decide how to distribute their armies between the $n$ battlefields, where each battlefield entails a reward to its winner, and in each battlefield, the outnumbering army wins. To the best of our knowledge, all-pay bidding games are the first to incorporate a modelling of bounded resource with no value for the players, as in Colonel Blotto games, with a dynamic behavior, as in ongoing auctions.

Graph games have been extensively used to model and reason about systems \cite{handbookMC} and multi-agent systems \cite{WG+16}. Bidding games naturally model systems in which the scheduler accepts payment in exchange for priority. Blockchain technology is one such system, where the {\em miners} accept payment and have freedom to decide the order of blocks they process based on the proposed transaction fees. Transaction fees are not refundable, thus all-pay bidding is the most appropriate modelling. Manipulating transaction fees is possible and can have dramatic consequences: such a manipulation was used to win a popular game on Ethereum called FOMO3d\footnote{\url{https://bit.ly/2wizwjj}}. There is thus ample motivation for reasoning and verifying blockchain technology \cite{CGV18}.

We show that all-pay bidding games exhibit an involved and interesting mathematical structure. As discussed in Example~\ref{ex:2-1}, while we show a complete characterization of the value function for the game $\WnR(2)$, it is significantly harder than the characterization with first-price bidding. The situation becomes worse when we slightly complicate the game and require that \PO wins three times in a row, called $\WnR(3)$, for short. We show that there are initial budgets for which an optimal strategy in $\WnR(3)$ requires infinite support. 

We turn to describe our positive results on general games. First, we study {\em surely-winning} in all-pay bidding games, i.e., winning with probability $1$. We show a threshold behavior that is similar to first-price bidding games: each vertex $v$ in a qualitative all-pay bidding game has a surely-winning threshold ratio, denoted $\SThresh(v)$, such that if \PO's ratio exceeds $\SThresh(v)$, he can guarantee winning the game with probability $1$, and if \PO's ratio is less than $\SThresh(v)$, \PT can guarantee winning with positive probability. Moreover, we show that surely-winning threshold ratios have the following structure: for every vertex $v$, we have $\SThresh(v) = \SThresh(v^-) + (1-\SThresh(v^-)/\SThresh(v^+))$, where $v^-$ and $v^+$ are the neighbors of $v$ that, respectively, minimize and maximize $\SThresh$. This result has computation-complexity implications; namely, we show that in general, the decision-problem counterpart of finding surely-winning threshold ratios is in PSPACE using the {\em existential theory of the reals} (ETR, for short) \cite{Can88}, and it is in linear time for DAGs. We show that surely-winning threshold ratios can be irrational, thus we conjecture that the decision problem is {\em sum-of-squareroot-hard}.

\begin{example}
\label{ex:tic-tac-toe}
Tic-tac-toe is a canonical graph game that is played on a DAG. First-price bidding Tic-tac-toe was discussed in a blog post\footnote{\url{https://bit.ly/2KUong4}}, where threshold budgets are shown to be surprisingly low: with Richman bidding, \PO can guarantee winning when his ratio exceeds $133/123 \approx 1.0183$ \cite{DP10} and with poorman-bidding, when it exceeds roughly $1.0184$. 

We implement our algorithm and find surely-winning threshold ratios in all-pay bidding tic-tac-toe: \PO surely wins when his initial ratio is greater than $51/31 \approx 1.65$ (see Fig.~\ref{fig:tic-tac-toe}). We point to several interesting phenomena. One explanation for the significant gap between the thresholds in all-pay and first-price bidding is that, unlike Richman and poorman bidding, in the range $(31/51, 51/31)$ neither player surely wins. Second, the threshold ratio for the relaxed \PO goal of surely not-losing equals the surely-winning threshold ratio. This is not always the case; e.g., from the configuration with a single $O$ in the middle left position, \PO requires a budget of $31/16$ for surely winning and $25/13$ for surely not-losing. Third, we find it surprising that when \PT wins the first bidding, it is preferable to set an $O$ in one of the corners (as in configuration $c_2$) rather than in the center.\hfill\qed
\begin{figure}[ht]
\center
\includegraphics[width=\linewidth]{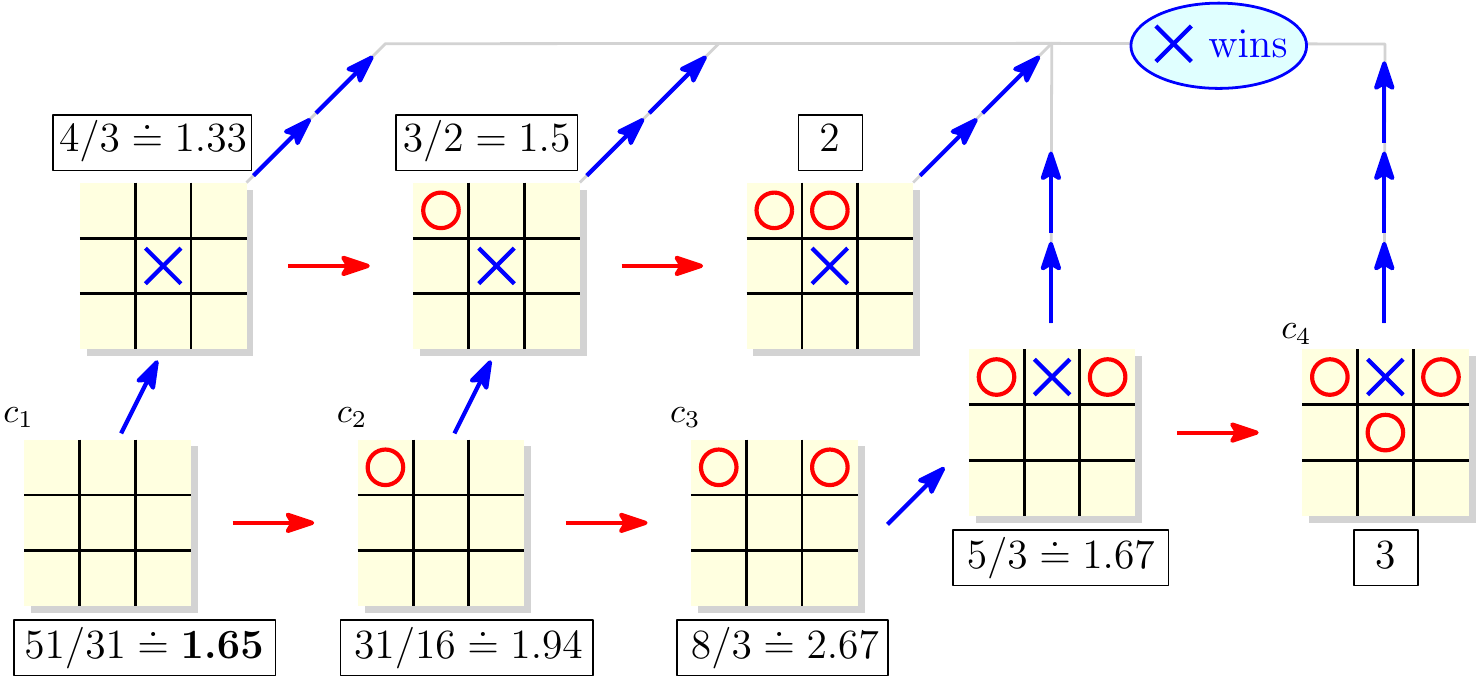}
\caption{Surely-winning thresholds in tic-tac-toe. For example, the surely-winning threshold budget in $c_{4}$ is $3$ since in order to win the game, \PO must win three biddings in a row.}
\label{fig:tic-tac-toe}
\end{figure}
\end{example}

Finally, we devise an FPTAS for the problem of finding values in DAGs; namely, given a bidding game $\G$ that is played on a DAG, an initial vertex, and an $\epsilon > 0$, we find, in time polynomial in the size of $\G$ and $1/\epsilon$, an upper- and lower-bound on the expected payoff that \PO can guarantee with every initial ratio of the form $\epsilon \cdot k$. The idea is to discretize the budgets and bids and, using a bottom-up approach, repeatedly solve finite-action zero-sum games. The algorithm gives theoretical bounds on the approximation. It is a simple algorithm that we have implemented and experimented with. Our experiments show that the difference between upper- and lower-bounds is small, thus we conclude that the algorithm supplies a good approximation to the value function. The experiments verify our theoretical findings. In addition, they hint at interesting behavior of all-pay bidding games, which we do not yet understand and believe encourages a further study of this model.

\subsection{Related work}
\label{sec:rel}

Colonel Blotto games \cite{Bor21} have been extensively studied; a handful of papers include \cite{Bel69,Bla54,Har08,SW81}. They have mostly been studied in the discrete case, i.e., the armies are given as individual soldiers, but, closer to our model, also in the continuous case (see \cite{Rob06} and references therein). The most well-studied objective is maximizing the expected payoff, though recently \cite{BB+19} study the objective of maximizing the probability of winning at least $k$ battlefields, which is closer to our model. 

All-pay bidding games were mentioned briefly in \cite{LLPSU99}, where it was observed that optimal strategies require probabilistic choices. To the best of our knowledge, all-pay bidding games \cite{MWX15} were studied only with discrete-bidding, which significantly simplifies the model, and in the Richman setting; namely, both players pay their bids to the other player. 

First-price bidding games have been well-studied and exhibit interesting mathematical properties. Threshold ratios in reachability Richman-bidding games, and only with these bidding rules, equal values in {\em random-turn based} games \cite{PSSW09}, which are stochastic games in which in each round, the player who moves is chosen according to a coin toss. This probabilistic connection was extended and generalized to infinite-duration bidding games with Richman- \cite{AHC19}, poorman- \cite{AHI18}, and even {\em taxman-bidding} \cite{AHZ19}, which generalizes both bidding rules. Other orthogonal extensions of the basic model include non-zero-sum bidding games \cite{MKT18} and discrete-bidding games that restrict the granularity of the bids \cite{DP10,AAH19}.

There are a number of shallow similarities between all-pay games and {\em concurrent stochastic games} \cite{Sha53}. A key difference between the models is that in all-pay bidding games, the (upper and lower) value depends on the initial ratio, whereas a stochastic game has one value. We list examples of similarities. In both models players pick actions simultaneously in each turn, and strategies require randomness and infinite memory though in {\em Everett recursive games} \cite{Ev55}, which are closest to our model, only finite memory is required and in more general stochastic games, the infinite memory requirement comes from remembering the history e.g. \cite{MN81} whereas in all-pay bidding games infinite support is already required in the game ``win 3 times in a row'' (which has histories of length at most $3$). Also, computing the value in stochastic games is in PSPACE using ETR  \cite{EKVY08}, it is in P for DAGs, and there are better results for solving the guaranteed-winning case \cite{CI15}.



\section{Preliminaries}
A reachability all-pay bidding game is $\zug{V, E, L, w}$, where $V$ is a finite set of vertices, $E \subseteq V \times V$ is a set of directed edges, $L \subseteq V$ is a set of {\em leaves} with no outgoing edges, and $w: L \rightarrow \Q$ assigns unique weights to leaves, i.e., for $l,l' \in L$, we have $w(l) \neq w(l')$. We require that every vertex in $V \setminus L$ has a path to at least two different leaves. A special case is {\em qualitative} games, where there are exactly two leaves with weights in $\set{0,1}$. We say that a game is played on a DAG when there are no cycles in the graph. For $v \in V$, we denote the {\em neighbors} of $v$ as $N(v) = \set{u \in V: \zug{v,u} \in E}$. 

A {\em strategy} is a recipe for how to play a game. It is a function that, given a finite {\em history} of the game, prescribes to a player which {\em action} to take,  where we define these two notions below. A history in a bidding game is $\pi = \zug{v_1, b^1_1, b^2_1}, \ldots, \zug{v_k, b^1_k, b^2_k}, v_{k+1} \in (V \times \Real \times \Real)^*\cdot V$, where for $1 \leq j \leq k+1$, the token is placed on vertex $v_j$ at round $j$ and \PLi's bid is $b^i_j$, for $i \in \set{1,2}$. Let $B^I_i$ be the initial budget of \PLi. \PLi's budget following $\pi$ is $B_i(\pi) = B^I_i - \sum_{1 \leq j \leq k} b^i_j$.  A play $\pi$ that ends in a leaf $l \in L$ is associated with the payoff $w(l)$. 

Consider a history $\pi$ that ends in $v \in V \setminus L$. The set of {\em legal} actions following $\pi$, denoted $A(\pi)$, consists of pairs $\zug{b, u}$, where $b \leq B_i(\pi)$ is a bid that does not exceed the available budget and $u \in N(v)$ is a vertex to move to upon winning. A {\em mixed} strategy is a function that takes $\pi$ and assigns a probability distribution over $A(\pi)$. The support of a strategy $f$ is $supp(f, \pi) = \set{\zug{b,u}: f(\pi)(\zug{v,u}) > 0}$. We assume WLog. that each bid is associated with one vertex to proceed to upon winning, thus if $\zug{b, v}, \zug{b, v'} \in supp(f, \pi)$, then $v=v'$. We say that $f$ is {\em pure} if, intuitively, it does not make probabilistic choices, thus for every history $\pi$, we have $|supp(f, \pi)|=1$. 

\begin{definition}[Budget Ratio]
Let $B_1, B_2 \in \Real$ be initial budgets for the two players. \PO's {\em ratio} is $B_1/B_2$.\footnote{We find this definition more convenient than $B_1/(B_1 + B_2)$ which is used in previous papers on bidding games.}
\end{definition}

Let $\G$ be an all-pay bidding game. An initial vertex $v_0$, an initial ratio $r \in \Real$, and two strategies $f_1$ and $f_2$ for the two players give rise to a probability $D(v_0, r, f_1, f_2)$ over plays, which is defined inductively as follows. The probability of the play of length $0$ is $1$. Let $\pi = \pi' \cdot \zug{v, b_1, b_2}, u$, where $\pi'$ ends in a vertex $v$. Then, we define $\Pr[\pi] = \Pr[\pi'] \cdot f_1(\pi)(b_1) \cdot f_2(\pi)(b_2)$. Moreover, for $i \in \set{1,2}$, assuming \PLi chooses the successor vertex $v_i$, i.e., $\zug{b_i, v_i} \in supp(f_i, \pi')$, then $u = v_1$ when $b_1 \geq b_2$ and otherwise $u=v_2$. That is, we resolve ties by giving \PO the advantage. This choice is arbitrary and does not affect most of our results, and the only affect is discussed in Remark~\ref{rem:tie}.

\begin{definition}[Game Values]
The {\em lower value} in an all-pay game $\G$ w.r.t. an initial vertex $v_0$, and an initial ratio $r$ is $val^\downarrow(\G, v_0, r) = \sup_f \inf_g \int_{\pi \in D(v_0, r, f, g)} \Pr[\pi] \cdot pay(\pi)$. The {\em upper value}, denoted $val^\uparrow(\G, v_0, r)$ is defined dually. It is always the case that $val^\downarrow(\G, v_0, r) \leq val^\uparrow(\G, v_0, r)$, and when $val^\uparrow(\G, v_0, r) = val^\downarrow(\G, v_0, r)$, we say that the value exists and we denote it by $val(\G, v_0, r)$. 
\end{definition}

\section{Surely-Winning Thresholds}
In this section we study the existence and computation of a necessary and sufficient initial budget for \PO that guarantees {\em surely} winning, namely winning with probability $1$. We focus on qualitative games in which each player has a target leaf. Note that the corresponding question in quantitative games is the existence of a budget that suffices to surely guarantee a payoff of some $c \in \Real$, which reduces to the the surely-winning question on qualitative games by setting \PO's target to be the leaves with weight at least $c$. We define surely-winning threshold budgets formally as follows.

\begin{definition}[Surely-Winning Thresholds]
Consider a qualitative game $\G$ and a vertex $v$ in $\G$. The surely-winning threshold at $v$, denoted $\SThresh(v)$, is a budget ratio such that: 
\begin{itemize}
\item If \PO's ratio exceeds $\SThresh(v)$, he has a strategy that guarantees winning with probability $1$.
\item If \PO's ratio is less than $\SThresh(v)$, \PT has a strategy that guarantees winning with positive probability.
\end{itemize}
\end{definition}

To show existence of surely-winning threshold ratios, we define {\em threshold functions}, show their existence, and show that they coincide with surely-winning threshold ratios.

\begin{definition}[Threshold functions]
\label{def:thresh}
\stam{
Consider a qualitative game $\G = \zug{V, E, t_1, t_2}$. Let $T:V \rightarrow \Real_{\geq 0}$ be a function. We call $T$ a {\em threshold function} if $T(t_1) = 0$, $T(t_2) = 1$, and for $v \in V \setminus \set{t_1,t_2}$, we have $T(v) = T(v^-) + 1-T(v^-)/T(v^+)$, where $v^-, v^+ \in N(V)$ are such that $T(v^-) \leq T(u) \leq T(v^+)$, for every $u \in N(v)$.
}
Consider a qualitative game $\G = \zug{V, E, t_1, t_2}$. Let $T:V \rightarrow \Real_{\geq 0}$ be a function. For $v \in V \setminus \set{t_1, t_2}$, let $v^-, v^+ \in N(V)$ be such that $T(v^-) \leq T(u) \leq T(v^+)$, for every $u \in N(v)$. We call $T$ a {\em threshold function} if
\[ 
T(v) = \begin{cases}
0 & \text{ if } v=t_1 \\
\infty & \text{ if } v = t_2\\
T(v^-) + 1-T(v^-)/T(v^+) & \text{ otherwise.}\footnotemark
\end{cases}
\]
\end{definition}

\footnotetext{We define $1/\infty = 0$ and $c < \infty$, for every $c \in \Real$.}
We start by making observations on surely-winning threshold ratios and threshold functions. 
\begin{lemma}
\label{lem:surely-obs}
Consider a qualitative game $\G = \zug{V, E, t_1, t_2}$, and let $T$ be a threshold function.
\begin{itemize}
\item For $v \in V$, if \PO's initial ratio exceeds $\SThresh(v)$, then he has a pure strategy that guarantees winning. 
\item For $v \in V \setminus \set{t_1, t_2}$, we have $\SThresh(v) \geq 1$. 
\item We have $T(v^-) \leq T(v) \leq T(v^+)$ and the inequalities are strict when $T(v^-) \neq T(v^+)$.
\end{itemize}
\end{lemma}
\begin{proof}
For the first claim, suppose \PO has a strategy $f$ that guarantees winning against any \PT strategy. Suppose $f$ is mixed. Then, we construct a pure strategy $f'$ by arbitrarily choosing, in each round, a bid $b \in supp(f)$. If \PT has a strategy $g$ that guarantees winning with positive probability against $f'$, then by playing $g$ against $f'$, he wins with positive probability, contradicting the assumption that $f'$ guarantees surely winning. For the second claim, suppose towards contradiction that $\SThresh(v) =1-\epsilon$, for $\epsilon > 0$. We think of \PO as revealing his pure winning strategy before \PT. Assuming \PO bids $b$ in a round, \PT reacts by bidding $b+\epsilon/|V|$. Thus, \PT wins $|V|$ biddings in a row and draws the game to $t_2$. A simple calculation verifies the last claim.
\end{proof}

We first show that threshold functions coincide with surely-winning threshold ratios, and then show their existence.

\begin{lemma}
\label{lem:surely-upper}
Consider a qualitative game $\G = \zug{V, E, t_1, t_2}$ and let $T$ be a threshold function for $\G$. Then, for every vertex $v \in V$, we have $T(v) \leq \SThresh(v)$; namely, \PO surely wins from $v$ when his budget exceeds $T(v)$.
\end{lemma}
\begin{proof}
We claim that if \PO's ratio is $T(v) + \epsilon$, for $\epsilon > 0$, he can surely-win the game. For $t_1$ and $t_2$, the claim is trivial and vacuous, respectively. We provide a strategy $f_1$ for \PO that guarantees that in at most $n=|V|$ steps, either \PO has won or he is at some node $u\in V$ with relative budget $T(u)+\epsilon+\deps$, where $\deps > 0$ is a small fixed positive number. By repeatedly applying this strategy, either \PO at some point wins directly, or he accumulates relative budget $n$ and then he can force a win in $n$ steps by simply bidding $1$ in each round. 

Suppose the token is placed on vertex $v\in V \setminus \set{t_1, t_2}$ following $0 \leq i \leq n$ biddings, let $v^-,v^+ \in N(v)$ be neighbors of $v$ that achieve the minimal and maximal value of $T$, respectively. For convenience, set $x = T(v^-)$, $y=T(v^+)$, and $\delta=\epsilon^2$.  \PO bids $1-x/y + \delta^{k+1-i}$. We first disregard the supplementary $\delta^{k+1-i}$. When \PO wins, we consider the worst case of \PT bidding $0$, thus \PO's normalized budget is greater than $\big(T(v) - (1-x/y)\big)/\big(1-(1-x/y)\big) = T(v^-)$. On the other hand, when \PO loses, \PT bids at least as much as \PO, and \PO's normalized budget is greater than $\big(T(v) - (1-x/y)\big)/\big(1-2(1-x/y)\big) = T(v^+)$. 

Upon winning a bidding, \PO moves to a neighbor $u \in N(v)$ with $T(u) = T(v^-)$, and when losing, the worst case is when \PT moves to a vertex $u$ with $T(u) = T(v^+)$. The claim above shows that in both cases \PO's budget exceeds $T(u)$. Since $T(t_2) = \infty$, we have established that the strategy guarantees not losing. 

We define \PO's moves precisely, and show how \PO uses the surplus $\delta$ to guarantee winning. We define \PO moves upon winning a bidding at $v \in V \setminus \set{t_1,t_2}$. He moves to $u \in N(v)$ such that $T(u)=T(v^-)$, where if there are several vertices that achieve the minimal value, he chooses a vertex that is closest to $t_1$. By Lemma~\ref{lem:surely-obs}, for every vertex $v$, we have $T(v^-) \leq T(v) \leq T(v^+)$ and equality occurs only when $T(v^-) = T(v^+)$. Thus, \PO's move guarantees that $T(u) \leq T(v)$ and if $u$ is farther from $t_1$ than $v$, then the inequality is strict. The sum of decreases in $T$ and in distance is at most $|V|$, thus $t_1$ is reached after winning $n$ biddings. 

We show how to use the surplus $\delta$ to guarantee winning. First note that $T(v) \geq 1-x/y$ and the extra terms $\delta^{k+1-i}$ add up to at most $\sum_{i=1}^k \delta^i < \sum_{i=1}^\infty \delta^i = \delta/(1-\delta)=\epsilon^2/(1-\epsilon^2)<\epsilon$ so the bids are legal (we assume, WLog., that $\epsilon \le 0.1$). Finally, if \PT wins the $i$-th bidding, \PO's ratio is at least $T(u)+ \epsilon - \frac{\delta^{k+1-i}}{1-\delta}$ whereas \PT's ratio is at most $x/y - \delta^{k+1-i}$. \PO's new ratio is then at least $\big( x+ \epsilon - \frac{\delta^{k+1-i}}{1-\delta} \big)/\big( x/y - \delta^{k+1-i} \big)$, and it is straightforward to check that this is at least as much as the desired $T(u)+\epsilon + \deps$, where $\deps=\epsilon^{2n}(1+\epsilon-1/(1-\epsilon^2))$.
\end{proof}

To obtain the converse, we show a strategy for \PT that wins with positive probability. We identify an upper bound $\beta$ in each vertex such that a \PO bid that exceeds $\beta$ exhausts too much of his budget. Then, \PT bids $0$ with probability $0.5$ and the rest of the probability mass is distributed uniformly in $[0, \beta]$. This definition intuitively allows us to reverse the quantification and assume \PO reveals his strategy before \PT; that is, when \PO bids at least $\beta$, we consider the case that \PT bids $0$, and when \PO bids less than $\beta$, we consider the case where \PT slightly overbids \PO. Both occur with positive probability.

\begin{lemma}
\label{lem:surely-lower}
Consider a qualitative game $\G = \zug{V, E, t_1, t_2}$ and let $T$ be a threshold function for $\G$. Then, for every vertex $v \in V$, we have $T(v) \geq \SThresh(v)$; namely, \PT wins with positive probability from $v$ when \PO's budget is less than $T(v)$.
\end{lemma}
\begin{proof}
We distinguish between two types of vertices in $V \setminus \set{t_1, t_2}$. The first type is $V_1 = \set{v: T(v^-) < T(v) < T(v^+)}$ and the second type is $V_2 = \set{v: T(v^-) = T(v) = T(v^+)}$. Assume \PT's budget is at least $1 + \epsilon$ and \PO's budget is at most $T(v) - \epsilon$. We use $\beta$ to denote an upper bound on \PT's bids. For $v \in V_1$, we set $\beta = 1-T(v^-)/T(v^+)$, and for $v \in V_2$, we set $\beta = p \cdot \epsilon (2n)^{-d}$, where $p$ is a small constant, $n = |V|$, and $d$ is the length of the shortest path from $v$ to a vertex in the first type of vertices. \PT bids $0$ with probability $1/2$ and the rest of the probability mass is distributed uniformly in $[0, 1]$. Upon winning, in the first case, \PT moves to a vertex $v^+$ that is closest to $t_2$, and in the second case, to a vertex that is closest to a vertex in $V_1$.

Intuitively, a bid above $\beta$ is too high for \PO since it exhausts too much of his budget. Thus, when \PO bids at least $\beta$, we consider the case where \PT bids $0$, and when \PO bids below $\beta$, we consider the case where \PT bids slightly above him, both occur with positive probability. Note that in both cases, we maintain the invariant that  \PT's budget is at least $1 + \epsilon$ and \PO's budget is at most $T(v) - \epsilon$, thus \PO does not win the game. We show that \PT wins with positive probability. Consider a finite play $\pi$ that is obtained from some choice of strategy of \PO. Note that a cycle in $\pi$ necessarily results from a \PO win (against a \PT bid of $0$), and such a win occurs when he bids at least $\beta$. In both cases, \PO's budget decreases by a constant. The tricky case is when the cycle is contained in $V_2$. Then, the sum of \PT's bids is at most $\sum_{i=k}^d p \epsilon (2n)^{-i} \leq (p \epsilon (2n)^{-d})/n$ whereas \PO bids at least $p \epsilon (2n)^{-d}$. Thus, \PO's decrease in budget is larger by a factor of $n$ of \PT. Thus, \PT's budget will eventually be larger than \PO's, at which point \PT wins in at most $n$ steps since he always bids above \PO's budget with positive probability.
\end{proof}

Finally, we show existence of threshold functions. We first show existence in DAGs, which is a simple backwards-induction argument that we have implemented (see Example~\ref{ex:tic-tac-toe}). To obtain existence in a general game $\G$, we find threshold functions in games on DAGs of the form $\G_n$, for $n \geq \Nat$, in which \PO is restricted to win in at most $n$ steps. We tend $n$ to infinity and show that the limit of the threshold functions in the games is a threshold function in $\G$.

\begin{lemma}
\label{lem:exist}
Every qualitative reachability bidding game has a threshold function.
\end{lemma}
\begin{proof}
Note that in a game that is played on a DAG, the existence of a threshold function is shown easily in a backwards-induction manner from the leaves. Consider a game $\G=\zug{V, E, t_1, t_2}$. For $n \in \Nat$ let $\G_n$ be the reachability game in which \PO wins iff the game reaches $t_1$ in at most $n$ rounds. Since $\G_n$ is a DAG, by the above, a threshold function $T_n$ exists in $\G_n$. By Lemma~\ref{lem:surely-upper}, we have $T_n(v) \geq T_{n+1}(v)$, and by Lemma~\ref{lem:surely-obs}, we have $T_n(v) \geq 1$, thus the sequence $\set{T_n(v)}_{n \geq 1}$ converges. We define $T(v) = \lim_{n \to \infty} T_n(v)$ and claim that $T$ is a threshold function. 

We show that $T(v)$ converges to $T(v^-) + 1- T(v^-)/T(v^+)$. Note that since $\G_n$ is a DAG, we have $T_n(v) = T_{n-1}(v^-) + 1- T_{n-1}(v^-)/T_{n-1}(v^+)$. Moreover, for $D = |V|$, note that $T_D(v) \leq D$, since a budget of $D$ allows \PO to win $D$ times in a row and draw the game to $t_1$. Since $\set{T_n(v)}_{n \geq 1}$ is a monotonically decreasing sequence, for every $\epsilon > 0$, there is $N \in \Nat$ such that for every $n > N$, we have $T(v) \leq T_n(v) \leq T(v) + \epsilon$. Given $\epsilon$, we choose $n$ such that 
\[
T(v) \leq T_n(v) \leq T(v^-) +\epsilon + 1- \frac{T(v^-)}{T(v^+)+\epsilon} \leq
\]
\[
\leq T(v^-) +\epsilon + 1- \frac{T(v^-)}{T(v^+)(1+\epsilon)} \leq 
\]
\[
\leq T(v^-) +\epsilon + 1- \frac{T(v^-)}{T(v^+)}(1-\epsilon)\leq
\]
\[
\leq T(v^-) +1- \frac{T(v^-)}{T(v^+)}+ \epsilon \cdot(D+1)
\]
\end{proof}

We use the characterization of surely-winning threshold ratios by means of threshold functions (Lemmas~\ref{lem:surely-upper} and~\ref{lem:surely-lower}) to find the threshold ratios in the game that is depicted in Fig.~\ref{fig:irrational}, and show that they can be irrational. In addition the characterization has computation-complexity consequences on the problem of finding surely-winning threshold ratios. In DAGs, finding threshold functions is done in linear time. In general graphs, the characterization implies a reduction to the {\em existential theory of the reals} (ETR, for short) by phrasing the constraints in Definition~\ref{def:thresh} as an input to ETR. It is known that ETR is in PSPACE \cite{Can88}. Combining the lemmas above, we have the following.

\begin{theorem}
\label{thm:surely}
Surely-winning threshold ratios exist in every qualitative reachability game and coincide with threshold functions. Surely-winning threshold ratios can be irrational. Given a vertex $v$ and a value $c \in \Q$, deciding whether $\SThresh(v) \geq c$ is in PSPACE for general graphs and can be solved in linear time for games on DAGs.
\end{theorem}

\begin{figure}
\center
\includegraphics[height=3cm]{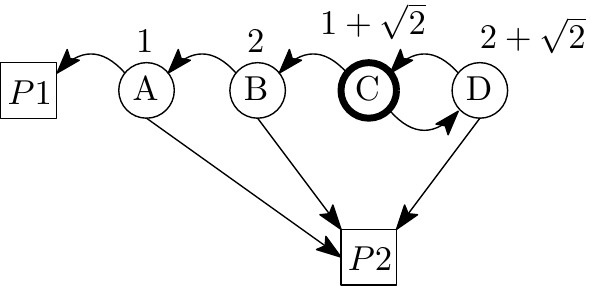}
\caption{A qualitative-reachability game with target $Pi$ for \PLi in which surely-winning threshold ratios are irrational.} 
\label{fig:irrational}
\end{figure}

\section{Finding Approximated Values}
\label{sec:approx}
In this section, we focus on games on DAGs. The algorithm that we construct is based on a discretization of the budgets; namely, we restrict the granularity of the bids and require \PO to bid multiples of some $\epsilon > 0$, similar in spirit to discrete-bidding games \cite{DP10}. We first relate the approximated value with the value in $\G$ with no restriction on the bids. Then, in Section~\ref{sec:exp}, we experiment with an implementation of the algorithm and show interesting behavior of all-pay bidding games. We define the approximate value as follows.

\begin{definition}[Approximate Value Function]
Let $\G$ be a game on a DAG and $\epsilon > 0$. Let $val_\epsilon$ be the {\em approximate-value} function in $\G$ when \PO is restricted to choose bids in $\set{\epsilon \cdot k :k \in \Nat}$ and \PT wins ties. 
\end{definition}

Our algorithm is based on the following theorem.

\begin{theorem}
\label{thm:approx-alg}
Consider a game on a DAG $\G$ where each leaf is labeled with a reward. Let $v$ be a vertex in $\G$, $B_1 \in \Real$ be an initial budget for \PO, $d(\G)$ be the longest path from a vertex to a leaf in $\G$, and  $\epsilon > 0$. Then, we have $val(v, B_1) \leq val_\epsilon(v, B_1 + d(\G)\cdot \epsilon)$.
\end{theorem}

Theorem~\ref{thm:approx-alg} gives rise to Algorithm~\ref{alg} that finds $val_\epsilon(B_1)$, for every $B_1$ that is a multiple of $\epsilon$. Note that assuming \PO bids only multiples of $\epsilon$, we can assume that \PT also bids multiples of $\epsilon$. The algorithm constructs a two-player zero-sum {\em matrix game}, which is $\zug{A_1, A_2, pay}$, where, for $i \in \set{1,2}$, $A_i$ is a finite set of actions for \PLi, and, given $a_i \in A_i$, the function $pay(a_1, a_2)$ is the payoff of the game. A solution to the game is the optimal payoff that \PO can guarantee with a mixed strategy, and it is found using linear programming. Let $\SThresh(v)$ denote the threshold budget with which \PO can guarantee the highest reward, then $val_\epsilon(B_1)$ equals the highest reward, for $B_1 > \SThresh(v)$. To compute $\SThresh(v)$, we use the linear-time algorithm in the previous section.

\begin{figure}
\begin{center}
\begin{algorithmic}
\algnotext{EndIf}
\algnotext{EndFor}
\State \underline{{\bf Approx-Values}($v, \epsilon$)}
\If {$N(v) = \emptyset$} \Return $weight(v)$ \EndIf
\State $\forall u \in N(v)$ call $\Call{Approx-Values}{u, \epsilon}$

\For {$B_1 = k \cdot \epsilon \text{ s.t. } 0 \leq B_1 \leq \SThresh(v)$} 
\State // Construct a two-player finite-action matrix game.
\State $A_1 = \set{b_1 = i \cdot \epsilon \text{ s.t. } 0 \leq b_1 \leq B_1}$
\State $A_2 = \set{b_2 = j \cdot \epsilon \text{ s.t. } 0 \leq b_2 \leq 1}$

	\For {$b_1 \in A_1, \ b_2 \in A_2$}
		\State // \PO's normalized new budget.
		\State $B'_1 = \lceil (B_1-b_1)/(1-b_2)\rceil_\epsilon$ 
		\If{$b_1 > b_2$} 
			\State $pay(b_1, b_2) = \max_{u \in N(v)} val_{\epsilon}(u, B'_1)$
		\Else
			\State $pay(b_1, b_2) = \min_{u \in N(v)} val_{\epsilon}(u, B'_1)$
		\EndIf
	\EndFor
	\State $val_\epsilon(v, B_1) = \Call{Solve}{A_1, A_2, pay}$
\EndFor
\State \Return $val_\epsilon$
\end{algorithmic}
\end{center}
\caption{An FPTAS for finding upper- and lower-bounds on the values for every initial budget ratio.}
\label{alg}
\end{figure}

\section{``Win $n$ in a Row'' Games}
In this section we study a simple fragment of qualitative games.
\begin{definition}[Win $n$ in a Row Games]
For $n \in \Nat$, let $\WnR(n)$ denote the qualitative game in which \PO needs to win $n$ biddings in a row and otherwise \PT wins. For example, see a depiction of $\WnR(2)$ in Fig.~\ref{fig:2-1}.
\end{definition}
We start with a positive results and completely solve $\WnR(2)$. Then, we show that optimal strategies require infinite support already in $\WnR(3)$.

\subsection{A solution to ``win twice in a row''}
We start by solving an open question that was posed in \cite{LLPSU99} and characterize the value as a function the budget ratio in the win twice in a row game $\WnR(2)$ (see a depiction of the game in Fig.~\ref{fig:2-1}).

\begin{theorem}
\label{thm:2-1}
Consider the all-pay bidding game $\WnR(2)$ in which \PO needs to win twice and \PT needs to win twice. The value exists for every pair of initial budgets. Moreover, suppose the initial budgets are $B_1$ for \PO and $1$ for \PT. Then, if $B_1 > 2$, the value is $1$, if $B_1 < 1$, the value is $0$, and if $B_1 \in (1+\frac{1}{n+1}, 1+\frac{1}{n}]$, for $n \in \Nat$, then the value is $\frac{1}{n+1}$.
\end{theorem}
\begin{proof}
The cases when $B_1 > 2$ and $B_1 < 1$ are easy. Let $2 \geq n \in \Nat$ such that $B_1 = 1+1/n + \epsilon$, where $\epsilon$ is such that $B_1 < 1+1/(n-1)$. We claim that the value of $\WnR(2)$ with initial budgets $B_1$ and $B_2 = 1$ is $1/n$. Consider the \PO strategy that choses a bid in $\set{k/n : 1 \leq k \leq n}$ uniformly at random. We claim that no matter how \PT bids, one of the choices wins, thus the strategy guarantees winning with probability at least $1/n$. Let $b_2$ be a \PT bid and let $k \in \Nat$ be such that $b_2 \in [k/n,(k+1)/n]$. Consider the case where \PO bids $b_1 = (k+1)/n$ and wins the first bidding. \PO's normalized budget in the second bidding is $B'_1 = \frac{B_1 - b_1}{1-b_2} \geq \frac{n+1/n + \epsilon - (k+1)/n}{1-(n-k)/n} > 1$, thus \PO wins the second bidding as well. Next, we show that \PO's strategy is optimal by showing a \PT strategy that guarantees winning with probability at least $(n-1)/n$. Let $\epsilon' > \epsilon$ such that $(n-1) \cdot \epsilon' < 1$, which exists since $B_1 < 1+1/(n-1)$. \PT chooses a bid uniformly at random in $\set{k \cdot (1/n + \epsilon'): 0 \leq k \leq n-1}$. Suppose \PO bids $b_1$, and we claim that the bid wins against at most one choice of \PT. Let $b_2$ be a \PT bid. When $b_2 > b_1$, \PT wins immediately. A simple calculation reveals that when $b_1 - b_2 > 1/n + \epsilon$, then \PO's normalized budget in the second bidding is less than $1$, thus he loses. It is not hard to see that there are $n-1$ choices for \PT that guarantee winning.
\end{proof}

\begin{remark}
\label{rem:tie}
Note that the tie-breaking mechanism affects the winner at the end-points of the intervals. For example, if we would have let \PT win ties, then the intervals would have changed to $[1+\frac{1}{n+1}, 1+\frac{1}{n})$.
\end{remark}

\subsection{Infinite support is required in $\WnR(3)$}
We continue to show a negative result already in $\WnR(3)$: there is an initial \PO budget for which his optimal strategy requires infinite support, which comes as a contrast to the optimal strategies we develop for $\WnR(2)$, which all have finite support.

\paragraph{A sweeping algorithm.}
In order to develop intuition for the proof, we present a sketch of an algorithm that decides, given $n \in \Nat$ and $B_1,v \in \Qpos$, whether \PO can guarantee winning with probability $v$ in $\WnR(n)$ when his initial budget is $B_1$ and \PT's budget is $1$. We assume access to a function $val_{n-1}$ that, given a budget $B_1$ for \PO, returns $val^\downarrow(B_1)$ in $\WnR(n-1)$. For example, Theorem~\ref{thm:2-1} shows that $val_2(1+1/m+\epsilon) = 1/m$, for $m \in \Nat$. The algorithm constructs a sequence of strategies $f_1, f_2, \ldots$ for \PO for the first bidding, each with finite support. We define $f_1(1) = v$ and $f_1(0) = 1-v$. That is, according to $f_1$, in the first bidding, \PO bids $1$ with probability $v$ and $0$ with probability $1-v$. Note that $out(f_1, 1) = v$. For $i \geq 1$, suppose $f_i$ is defined, that its support is $1=b_1 > \ldots > b_i> 0$, and that $out(f_i, b) \geq v$, for any \PT bid $b \geq b_i$. Intuitively, as \PT lowers his bid, his remaining budget for the subsequent biddings increases and the winning probability for \PO decreases. We ``sweep down'' from $b_i$ and find the maximal bid $b_{i+1}$ of \PT such that $out(f_i, b_{i+1}) < v$. We define $f_{i+1}$ by, intuitively, shifting some of the probability mass that $f_i$ assigns to $0$ to $b_{i+1}$, thus $supp(f_{i+1}) = supp(f_i) \cup  \set{b_{i+1}}$ and $f_i(b_j) = f_{i+1}(b_j)$, for $1 \leq j \leq i$. 
We terminate in two cases. If there is no positive $b$ such that $out(f_i, b) < v$, then $f_i$ guarantees a value of $v$, and we return it. On the other hand, if there is no $f_i(b) \in [0,1]$ that guarantees $out(f_i, b) = v$, then there $val(B_1) < v$. For example, in the game $\WnR(2)$ with $B_1 = 1\frac{1}{3}$ and $v = 1/3$, we define $f_1(1) = 1/3$. We find $b_2$ by solving $(B_1 - b_1)/(1 - b_2)=1$ to obtain $b_2 = 2/3$, and define $f_2(b_2) = 1/3$. Similarly, we find $b_3$ by solving $(B_1 - 2/3)/(1-b_3) = 1$, and terminate.

\paragraph{Infinite support is required.}
The following lemma, shows a condition for the in-optimality of a strategy. Intuitively, consider a \PO strategy $f$ that is obtained as in the sweeping algorithm only that in the $(i+1)$-th iteration, instead of adding $x < b_i$ to the support, it adds $b_{i+1} < x$. The value that $f$ guarantees needs to be at least $v$ for any \PT bid, and specifically for $b_i$ and $x$. Note that it is clearly beneficial for \PO to bid $x$ rather than $b_i$ against the bid $x$ of \PT. Since $x$ is not in the support of $f$, the probability $f(b_i)$ is un-necessarily high to ``compensate''. We construct $f'$ by shifting some of this probability to $x$, we are left with a ``surplus'' which we re-distribute to the rest of the support of $f$, thus $f'$ guarantees a better value than $f$. Formally, we have the following.

\begin{lemma}
\label{lem:non-opt}
Consider a \PO strategy $f$ in $\WnR(n)$, for some $n\in \Nat$, that has finite support $b_1 > \ldots> b_m$ in the first bidding. If there is $1 \leq i < m$, $1 \leq k \leq i$, and $b_i > x > b_{i+1}$ with $out(b_k, b_i) > out(b_k, x)$ and $out(x,x) = out(b_i, b_i)$, then $f$ is not optimal.
\end{lemma}
\begin{proof}
Suppose that the strategy $f$ guarantees a value of $v$ and let $x \in \Real$ be as in the statement. The value is at least $v$ against any \PT strategy, and in particular against the bids $b_i$ and $x$, thus $out(f, x) \geq v$. We claim that $out(f, x) < out(f, b_i)$, thus $out(f, b_i) > v$. Recall that the definition of the outcome is $out(f, b_i) = \sum_{1 \leq j \leq i} f(b_j) \cdot out(b_j, b_i)$ and similarly for $out(f, x)$. Fixing a \PO bid, the function $out$ is monotonically decreasing with \PT's bids since the lower \PT bids, the more budget he has left for subsequent biddings. Thus, for $1 \leq j \leq i$, we have $out(b_j, b_i) \geq out(b_j, x)$, and the assumptions of the lemma imply that at least one of the inequalities is strict. Thus, the claim follows and we have $out(f, b_i) > v$.

Next, we construct a ``partial-strategy'' $f'$ by adding $x$ to the support. Formally, we define $f'(b_j) = f(b_j)$, for $1 \leq j < i$, and let $f'(b_i)$ be such that $\sum_{1 \leq j \leq i} f'(b_j) \cdot out(b_j, b_i) = v$. Let $b_{i+1} = x$, and let $f'(b_{i+1})$ be such that $\sum_{1 \leq j \leq i+1} f'(b_j) \cdot out(b_j, b_i) = v$. We claim that, intuitively, we are left with a ``surplus'', and formally, we have $f'(b_i) + f'(b_{i+1}) < f(b_i)$. As in the above, we have $out(f, b_i) = \sum_{1 \leq j \leq i} f(b_j) \cdot out(b_j, b_i) > v$. Suppose we chose the maximal $b_i > x > b_{i+1}$ for which the statement holds, thus for every $j \neq k$, we have $out(b_j, x) = out(b_j, b_i)$, and $out(b_k, x) = out(b_k, b_i) - c$, for some $c >0$. We subtract the equality $\sum_{1 \leq j \leq i+1} f'(b_j) \cdot out(b_j, b_i) = v$ from $out(f, b_i)  > v$ and plug in the equalities to obtain $f'(b_i) \cdot out(b_i, b_i) + f'(b_{i+1}) \cdot out(x, x) < f(b_i) \cdot out(b_i, b_i) - c \cdot f(b_k)$. Since we assume $out(x, x) = out(b_i,b_i)$, the claim follows. 

Let $\Delta = f(b_i) - (f'(b_i) + f'(b_{i+1}))$ denote our ``surplus'', which, by the above, is positive. We define a new \PO strategy $f''$ with support $supp(f) \cup \set{x}$. The probability of $b_j$, for $j \neq i$ is $f(b_j)+ \Delta/|supp(f'')|$. The probability of $b_i$ is $f'(b_i) + \Delta/|supp(f'')|$ and the probability of $x$ is $f'(b_{i+1}) +\Delta/|supp(f'')|$. It is not hard to show that $f''$ guarantees a value that is greater than $v$, thus $f$ is not optimal. 
\end{proof}

Next, we use Lemma~\ref{lem:non-opt} to show that any strategy with finite support is not optimal.

\begin{theorem}
\label{thm:3-1}
Consider the game $\WnR(3)$ in which \PO needs to win three biddings and \PT needs to win once. Suppose the initial budgets are $1.25$ for \PO and $1$ for \PT. Then, an optimal strategy for \PO requires infinite support in the first bidding.
\end{theorem}
\begin{proof}
Suppose towards contradiction that $f$ is an optimal strategy with finite support in the first bidding. Consider the infinite sequence $x_n = 7/8 - 1/8 \cdot \sum_{1 \leq j \leq n-1} 2^{-j}$, for $n \geq 1$. It is not hard to verify that $(5/4 - x_n)/(1-x_{n+1}) = 2$, for every $n \geq 1$. That is, when \PO bids $x_n$ and \PT bids $x_{n+1}$, by Theorem~\ref{thm:2-1}, the value is $0.5$.

Let $supp(f) = \set{b_1,\ldots,b_m}$ be the support of $f$ in the first bidding, where $b_1 > b_2 > \ldots > b_m$. Since the support is finite, there is an $1 \leq i \leq m$ and $k \geq 1$ such that $b_i \geq x_k > x_{k+1} > b_{i+1}$. We claim that $x_{k+1}$ satisfies the assumptions of Lemma~\ref{lem:non-opt}, thus $f$ is not optimal. First, it is not hard to verify that for $0.75 < y < 1$, we have $(5/4-y)/(1-y) > 2$. Since the sequence $\set{x_n}_{n\geq 1}$ tends to $0.75$, we have $out(b_i, b_i) = out(x_{k+1}, x_{k+1}) = 1$. Second, when \PT's bid is fixed, then $out$ is monotonically decreasing with \PO's bid, thus $out(b_i, x_{k+1}) \leq out(x_k, x_{k+1}) = 0.5$. Thus, $out(b_i, b_i) > out(b_i, x_{k+1})$, and we are done.
\end{proof}

\section{Experiments}
\label{sec:exp}
\begin{figure}
\center
\includegraphics[height=5.5cm]{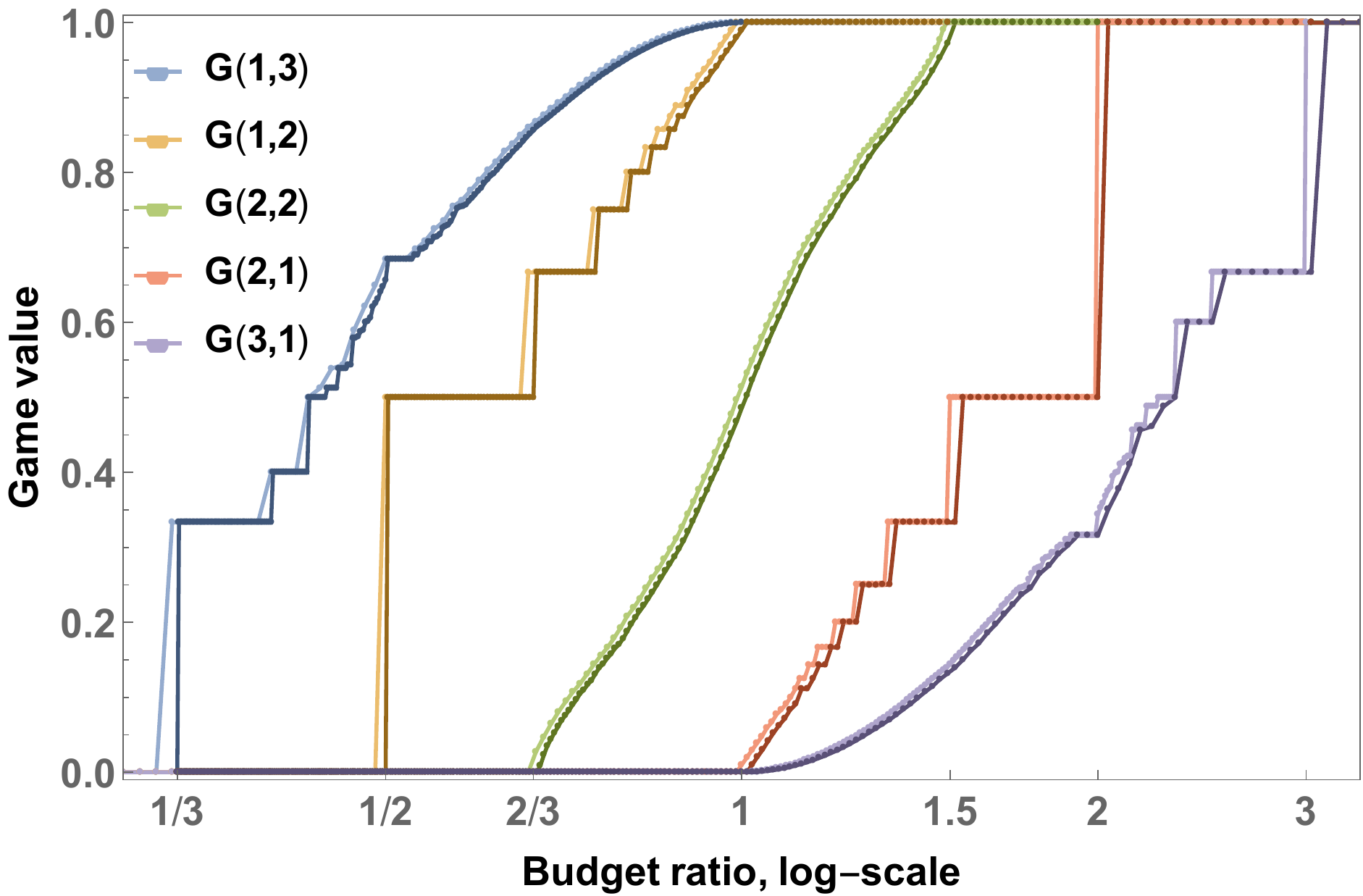}
\caption{Upper- and lower-bounds on the values of the games $G(1,3), G(1,2), G(2,2), G(2,1)$, and $G(3,1)$.} 
\label{fig:plot-1}
\end{figure}
We have implemented Algorithm~\ref{alg} and experiment by running it on qualitative games that are called a {\em race} in \cite{HV85}.
\begin{definition}[Race Games]
\label{def:race}
For $n, m \in \Nat$, let $G(n,m)$ be the qualitative game that consists of at most $n+m-1$ biddings in which \PO wins if he wins $n$ biddings and \PT wins if he wins $m$ biddings. Specifically, we have $\WnR(n) = G(n,1)$. 
\end{definition}
The algorithm is implemented in Python and it is run on a personal computer. In our experiments, we choose $\epsilon = 0.01$. In terms of scalability, the running time for for solving the game $G(5,5)$ is at most $10$ minutes. Solutions to smaller games are in the order of seconds to several minutes.

Figure~\ref{fig:plot-1} depicts some values for five games as output by the implementation. We make several observations. First, close plots represent the upper- and lower-bounds of a game. We find it surprising that the difference between the two approximations is very small, and conclude that the output of the algorithm is a good approximation of the real values of the games. Second, the plot of $G(2,1) = \WnR(2)$ (depicted in red), is an experimental affirmation of Theorem~\ref{thm:2-1}; namely, the step-wise behavior and the values that are observed in theory are clearly visible in the plot. Third, while a step-wise behavior can be seen for high initial budgets in $G(3,1)$ (depicted in purple), for lower initial budgets, the behavior seems more involved and possibly continuous. In Theorem~\ref{thm:3-1}, we show that optimal strategies for initial budgets in this region require infinite support, and the plot affirms the more involved behavior of the value function that is predicted by the theorem. Continuity is more evident in more elaborate games whose values are depicted in Figure~\ref{fig:plot-2}. Both plots give rise to several interesting open questions, which we elaborate on in the next section.
\begin{figure}
\center
\includegraphics[height=5.5cm]{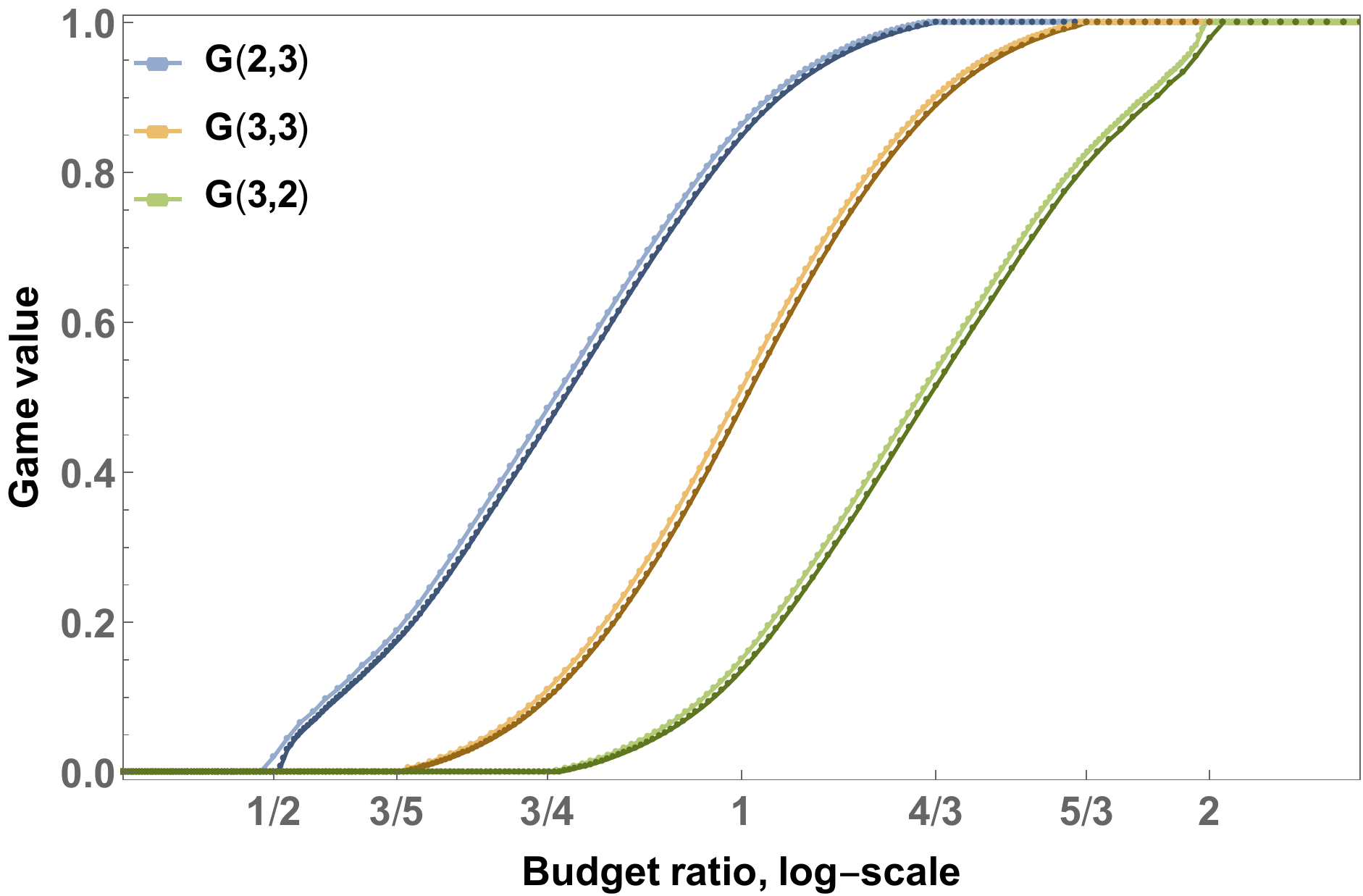}
\caption{Upper- and lower-bounds on the values of the games $G(3,2), G(3,3)$, and $G(2,3)$.} 
\label{fig:plot-2}
\end{figure}
\begin{figure}
\center
\includegraphics[height=5.5cm]{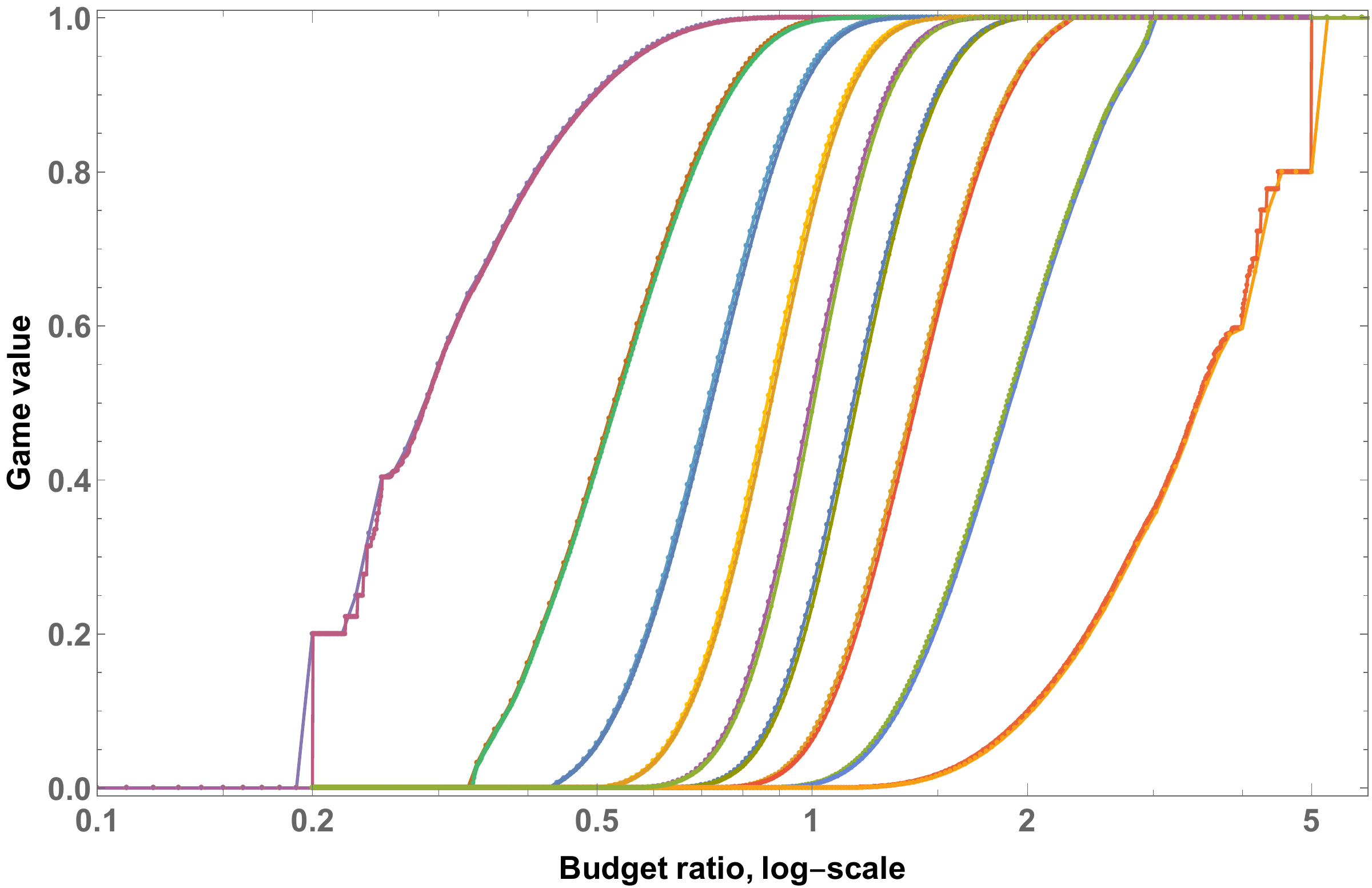}
\caption{Upper- and lower-bounds on the values of the games $G(5,i)$ and $G(i,5)$, for $1 \leq i \leq 5$.} 
\label{fig:plot-3}
\end{figure}

\section{Discussion}
We study, for the first time, all-pay bidding games on graphs. Unlike bidding games that use variants of first-price auctions, all-pay bidding games appropriately model decision-making settings in which bounded effort with no inherent value needs to be invested dynamically. While our negative results show that all-pay bidding games are significantly harder than first-price bidding games, our results are mostly positive. We are able to regain the threshold-ratio phenomena from first-price bidding games by considering surely-winning threshold ratios, and our implementation for games on DAGs solves non-trivial games such as tic-tac-toe. We show a simple FPTAS that finds upper- and lower-bounds on the values for every initial budget ratio, which we have implemented and show that it performs very well. 

We leave several open questions. The basic question on games, which we leave open, is showing the existence of a value with respect to every initial ratio in every game. We were able to show existence in $\WnR(2)$, and Fig.~\ref{fig:plot-2} hints the value exists in more complicated games. Also, while we identify the value function in $\WnR(2)$, we leave open the problem of a better understanding of this function in general. For example, while for $\WnR(2)$ we show that it is a step-wise function, observing Fig.~\ref{fig:plot-2} it seems safe to guess that the function can be continuous. Finally, characterizing the function completely, similar to our solution of $\WnR(2)$, in more involved games, is an interesting and challenging open problem.

\section{Acknowledments}
This research was supported by the Austrian Science Fund (FWF) under grants S11402-N23 (RiSE/SHiNE), Z211-N23 (Wittgenstein Award), and M 2369-N33 (Meitner fellowship).

\begin{small}
\bibliographystyle{aaai}
\bibliography{../ga}

\begin{thebibliography}{}

\bibitem[\protect\citeauthoryear{Aghajohari, Avni, and Henzinger}{2019}]{AAH19}
Aghajohari, M.; Avni, G.; and Henzinger, T.~A.
\newblock 2019.
\newblock Determinacy in discrete-bidding infinite-duration games.
\newblock In {\em Proc. 30th CONCUR}, volume 140 of {\em LIPIcs},  20:1--20:17.

\bibitem[\protect\citeauthoryear{Avni, Henzinger, and Chonev}{2019}]{AHC19}
Avni, G.; Henzinger, T.~A.; and Chonev, V.
\newblock 2019.
\newblock Infinite-duration bidding games.
\newblock {\em J. ACM} 66(4):31:1--31:29.

\bibitem[\protect\citeauthoryear{Avni, Henzinger, and
  Ibsen-Jensen}{2018}]{AHI18}
Avni, G.; Henzinger, T.~A.; and Ibsen-Jensen, R.
\newblock 2018.
\newblock Infinite-duration poorman-bidding games.
\newblock In {\em Proc. 14th WINE}, volume 11316 of {\em LNCS},  21--36.
\newblock Springer.

\bibitem[\protect\citeauthoryear{Avni, Henzinger, and
  \v{Z}ikeli\'c}{2019}]{AHZ19}
Avni, G.; Henzinger, T.~A.; and \v{Z}ikeli\'c, {\DJ}.
\newblock 2019.
\newblock Bidding mechanisms in graph games.
\newblock In {\em In Proc. 44th MFCS}, volume 138 of {\em LIPIcs},
  11:1--11:13.

\bibitem[\protect\citeauthoryear{Baye and Hoppe}{2003}]{BH03}
Baye, M.~R., and Hoppe, H.~C.
\newblock 2003.
\newblock The strategic equivalence of rent-seeking, innovation, and
  patent-race games.
\newblock {\em Games and Economic Behavior} 44(2):217--226.

\bibitem[\protect\citeauthoryear{Behnezhad \bgroup et al\mbox.\egroup
  }{2019}]{BB+19}
Behnezhad, S.; Blum, A.; Derakhshan, M.; Hajiaghayi, M.~T.; Papadimitriou,
  C.~H.; and Seddighin, S.
\newblock 2019.
\newblock Optimal strategies of blotto games: Beyond convexity.
\newblock In {\em Proc. the 20th EC},  597--616.

\bibitem[\protect\citeauthoryear{Bellman}{1969}]{Bel69}
Bellman, R.
\newblock 1969.
\newblock On ``colonel blotto'' and analogous games.
\newblock {\em SIAM Rev.} 11(1):66--68.

\bibitem[\protect\citeauthoryear{Blackett}{1954}]{Bla54}
Blackett, D.~W.
\newblock 1954.
\newblock Some blotto games.
\newblock {\em NRL} 1(1):55--60.

\bibitem[\protect\citeauthoryear{Borel}{1921}]{Bor21}
Borel, E.
\newblock 1921.
\newblock La th\'eorie du jeu les \'equations int\'egrales \'a noyau
  sym\'etrique.
\newblock {\em Comptes Rendus de l'Acad\'emie} 173(1304--1308):58.

\bibitem[\protect\citeauthoryear{Canny}{1988}]{Can88}
Canny, J.~F.
\newblock 1988.
\newblock Some algebraic and geometric computations in {PSPACE}.
\newblock In {\em Proc. 20th STOC},  460--467.

\bibitem[\protect\citeauthoryear{Chatterjee and Ibsen{-}Jensen}{2015}]{CI15}
Chatterjee, K., and Ibsen{-}Jensen, R.
\newblock 2015.
\newblock The value 1 problem under finite-memory strategies for concurrent
  mean-payoff games.
\newblock In {\em Proc. 26th SODA},  1018--1029.

\bibitem[\protect\citeauthoryear{Chatterjee, Goharshady, and
  Velner}{2018}]{CGV18}
Chatterjee, K.; Goharshady, A.~K.; and Velner, Y.
\newblock 2018.
\newblock Quantitative analysis of smart contracts.
\newblock In {\em Proc. 27th ESOP},  739--767.

\bibitem[\protect\citeauthoryear{Chatterjee, Reiter, and Nowak}{2012}]{CRN12}
Chatterjee, K.; Reiter, J.~G.; and Nowak, M.~A.
\newblock 2012.
\newblock Evolutionary dynamics of biological auctions.
\newblock {\em Theoretical Population Biology} 81(1):69 -- 80.

\bibitem[\protect\citeauthoryear{Clarke \bgroup et al\mbox.\egroup
  }{2018}]{handbookMC}
Clarke, E.~M.; Henzinger, T.~A.; Veith, H.; and Bloem, R., eds.
\newblock 2018.
\newblock {\em Handbook of Model Checking}.
\newblock Springer.

\bibitem[\protect\citeauthoryear{Develin and Payne}{2010}]{DP10}
Develin, M., and Payne, S.
\newblock 2010.
\newblock Discrete bidding games.
\newblock {\em The Electronic Journal of Combinatorics} 17(1):R85.

\bibitem[\protect\citeauthoryear{Etessami \bgroup et al\mbox.\egroup
  }{2008}]{EKVY08}
Etessami, K.; Kwiatkowska, M.~Z.; Vardi, M.~Y.; and Yannakakis, M.
\newblock 2008.
\newblock Multi-objective model checking of markov decision processes.
\newblock {\em LMCS} 4(4).

\bibitem[\protect\citeauthoryear{Everett}{1955}]{Ev55}
Everett, H.
\newblock 1955.
\newblock Recursive games.
\newblock {\em Annals of Mathematics Studies} 3(39):47--78.

\bibitem[\protect\citeauthoryear{Harris and Vickers}{1985}]{HV85}
Harris, C., and Vickers, J.
\newblock 1985.
\newblock Perfect equilibrium in a model of a race.
\newblock {\em The Review of Economic Studies} 52(2):193 -- 209.

\bibitem[\protect\citeauthoryear{Hart}{2008}]{Har08}
Hart, S.
\newblock 2008.
\newblock Discrete colonel blotto and general lotto games.
\newblock {\em International Journal of Game Theory} 36(3):441--460.

\bibitem[\protect\citeauthoryear{Klumppa and K.Polborn}{2006}]{KP06}
Klumppa, T., and K.Polborn, M.
\newblock 2006.
\newblock Primaries and the new hampshire effect.
\newblock {\em Journal of Public Economics} 90(6--7):1073--1114.

\bibitem[\protect\citeauthoryear{Konrad and Kovenock}{2009}]{KK09}
Konrad, K.~A., and Kovenock, D.
\newblock 2009.
\newblock Multi-battle contests.
\newblock {\em Games and Economic Behavior} 66(1):256--274.

\bibitem[\protect\citeauthoryear{Lazarus \bgroup et al\mbox.\egroup
  }{1996}]{LLPU96}
Lazarus, A.~J.; Loeb, D.~E.; Propp, J.~G.; and Ullman, D.
\newblock 1996.
\newblock Richman games.
\newblock {\em Games of No Chance} 29:439--449.

\bibitem[\protect\citeauthoryear{Lazarus \bgroup et al\mbox.\egroup
  }{1999}]{LLPSU99}
Lazarus, A.~J.; Loeb, D.~E.; Propp, J.~G.; Stromquist, W.~R.; and Ullman, D.~H.
\newblock 1999.
\newblock Combinatorial games under auction play.
\newblock {\em Games and Economic Behavior} 27(2):229--264.

\bibitem[\protect\citeauthoryear{Meir, Kalai, and Tennenholtz}{2018}]{MKT18}
Meir, R.; Kalai, G.; and Tennenholtz, M.
\newblock 2018.
\newblock Bidding games and efficient allocations.
\newblock {\em Games and Economic Behavior}.

\bibitem[\protect\citeauthoryear{Menz, Wang, and Xie}{2015}]{MWX15}
Menz, M.; Wang, J.; and Xie, J.
\newblock 2015.
\newblock Discrete all-pay bidding games.
\newblock {\em CoRR} abs/1504.02799.

\bibitem[\protect\citeauthoryear{Mertens and Neyman}{1981}]{MN81}
Mertens, J., and Neyman, A.
\newblock 1981.
\newblock Stochastic games.
\newblock {\em International Journal of Game Theory} 10(2):53--66.

\bibitem[\protect\citeauthoryear{Michael R.~Baye and de Vries}{1996}]{BKV96}
Michael R.~Baye, D.~K., and de~Vries, C.~G.
\newblock 1996.
\newblock The all-pay auction with complete information.
\newblock {\em Economic Theory} 8(2):291--305.

\bibitem[\protect\citeauthoryear{Peres \bgroup et al\mbox.\egroup
  }{2009}]{PSSW09}
Peres, Y.; Schramm, O.; Sheffield, S.; and Wilson, D.~B.
\newblock 2009.
\newblock Tug-of-war and the infinity laplacian.
\newblock {\em J. Amer. Math. Soc.} 22:167--210.

\bibitem[\protect\citeauthoryear{Roberson}{2006}]{Rob06}
Roberson, B.
\newblock 2006.
\newblock The colonel blotto game.
\newblock {\em Economic Theory} 29(1):1--24.

\bibitem[\protect\citeauthoryear{Shapley}{1953}]{Sha53}
Shapley, L.~S.
\newblock 1953.
\newblock Stochastic games.
\newblock {\em Proceedings of the National Academy of Sciences}
  39(10):1095--1100.

\bibitem[\protect\citeauthoryear{Shubik and Weber}{1981}]{SW81}
Shubik, M., and Weber, R.~J.
\newblock 1981.
\newblock Systems defense games: Colonel blotto, command and control.
\newblock {\em NRL} 28(2):281 -- 287.

\bibitem[\protect\citeauthoryear{Tullock}{1980}]{Tul80}
Tullock, G.
\newblock 1980.
\newblock {\em Toward a Theory of the Rent Seeking Society}.
\newblock College Station: Texas A\&M Press.
\newblock chapter Efficient rent seeking,  97--112.

\bibitem[\protect\citeauthoryear{Wooldridge \bgroup et al\mbox.\egroup
  }{2016}]{WG+16}
Wooldridge, M.~J.; Gutierrez, J.; Harrenstein, P.; Marchioni, E.; Perelli, G.;
  and Toumi, A.
\newblock 2016.
\newblock Rational verification: From model checking to equilibrium checking.
\newblock In {\em Proc. of the 30th AAAI},  4184--4191.

\end{thebibliography}
\end{small}
\end{document}